\documentclass[12pt]{article}

\usepackage{amsmath,amssymb,amsfonts}
\usepackage{amsthm}
\usepackage{algorithm}
\usepackage{algorithmic}
\usepackage{booktabs}
\usepackage{array}
\usepackage{multirow}
\usepackage{url}
\usepackage{graphicx}
\usepackage{float}
\usepackage{placeins}
\usepackage{geometry}
\usepackage{hyperref}
\usepackage{adjustbox}

\newtheorem{theorem}{Theorem}
\newtheorem{lemma}{Lemma}
\newtheorem{proposition}{Proposition}
\newtheorem{corollary}{Corollary}
\newtheorem{definition}{Definition}
\newtheorem{remark}{Remark}
\newtheorem{example}{Example}

\newcommand{\Ht}{\mathcal{H}_t}
\newcommand{\Bt}{\mathcal{B}_t}
\newcommand{\D}{\mathcal{D}}
\newcommand{\Fv}{F_V}
\newcommand{\Fe}{F_E}
\newcommand{\T}{\mathcal{T}}
\newcommand{\Cset}{\mathcal{C}}
\newcommand{\Kcomp}{\mathcal{K}}
\newcommand{\dist}{\operatorname{dist}}
\newcommand{\can}{\operatorname{can}}

\newcommand{\ecc}{\operatorname{ecc}}

\newcommand{\rhoEJ}{\rho}

\title{Re-Rooting-Assisted Edge-Minimum Runtime Repair for Node and Link Failures in Dense Eisenstein--Jacobi Broadcast Networks}

\author{Bader A. Albader\\
\small Department of Computer Science, Faculty of Science, Kuwait University, Kuwait\\
\small \texttt{albader@cs.ku.edu.kw}}

\date{}

\begin{document}
\maketitle

\begin{abstract}
One-to-all broadcasting in dense Eisenstein--Jacobi (EJ) networks relies on diameter-level spanning trees that fragment when nodes or links fail. This paper introduces the selected triple $(r,\theta,\Kcomp_{r,\theta})$--a chosen root, a chosen EJ coordinate-reduction orientation, and the healthy component graph induced by that choice--as the fundamental unit of analysis for joint node/link fault recovery. The central result is a necessary and sufficient condition: hybrid repair succeeds if and only if the healthy EJ graph $G'=\Ht-\Fv-\Fe$ is connected. When $G'$ is connected, a spanning tree of $\Kcomp_{r,\theta}$ maps to exactly $c-1$ component-crossing repair edges, which is minimum for the selected pruned tree. Deterministic guarantees include: one/two faulty nodes are always placed on the distance-$t$ boundary by re-rooting; a single failed link is either avoided or repaired by exactly one crossing edge; and the repaired depth satisfies $D_{r,\theta}\le 2t+1$ under shallowest-layer entry selection. A 260,000-trial validation campaign confirms 100\% recovery and substantial repair-edge reduction over fixed-source repair across five network scales up to $N=120601$ nodes, while global-BFS, near-miss, and cap-sensitivity audits clarify the tradeoff between reachability, forwarding-state changes, and ranked root selection.
\end{abstract}


\section{Introduction}
One-to-all broadcasting is a core collective operation in parallel computers, many-core systems, Network-on-Chip fabrics, and distributed accelerators. Dense Eisenstein--Jacobi (EJ) networks are attractive for such communication because they have degree six, small diameter, vertex symmetry, and a regular hexagonal coordinate geometry. The fault-free EJ broadcast reaches all nodes in one network diameter by expanding through six directional sectors. However, if a faulty node or failed link appears at an internal forwarding position, the fault-free tree can fragment and downstream nodes may fail to receive the broadcast.

Two existing EJ mechanisms motivate the present work. First, source re-rooting can relocate the effective source so that one or two faulty nodes lie on the graph-distance-$t$ boundary of the new source; in the standard EJ broadcast, boundary nodes are leaves and do not forward. Second, component repair can reconnect a fault-pruned EJ broadcast tree by contracting its healthy components and adding component-crossing edges. These mechanisms solve different problems. Re-rooting alone does not handle failed links or residual component cuts. Component repair alone repairs a chosen tree, but does not explain how to choose a source and orientation that reduce the damage before repair.

This paper studies the combined object
\begin{equation}
(r,\theta,\Kcomp_{r,\theta}),
\end{equation}
where $r$ is the selected effective source, $\theta$ is a selected EJ coordinate-reduction orientation, and $\Kcomp_{r,\theta}$ is the healthy component graph induced after node and link pruning. The optimization is therefore not only to find a valid re-rooted source or to repair a fixed tree. It is to choose a root--orientation pair that reduces the residual component structure and then repair that structure with the minimum possible number of external component-crossing edges.

\subsection{Novelty Separation}
Table~\ref{tab:novelty-separation} separates this paper from the two direct EJ ingredients. This separation is important because the present paper should not be read as merely repeating the one/two-node boundary theorem or the $c-1$ component theorem. Those are foundations. The new problem is joint node/link/runtime recovery by selecting the coupled unit $(r,\theta,\Kcomp_{r,\theta})$.

\begin{table}[H]
\centering
\caption{Novelty separation from the two EJ foundation mechanisms.}
\label{tab:novelty-separation}
\begin{adjustbox}{max width=\textwidth}
\begin{tabular}{p{0.22\linewidth}p{0.34\linewidth}p{0.34\linewidth}}
\toprule
Ingredient & Solves & Not solved by that ingredient alone \\
\midrule
EJ re-rooting & Places one/two faulty nodes on the distance-$t$ boundary so they become leaves in the standard broadcast. & Failed links, runtime-discovered links, mixed node/link cuts, residual component fragmentation, and repair-edge minimization. \\
EJ-MOEM component repair & For a chosen EJ tree whose healthy component graph is connected, repairs $c$ components using exactly $c-1$ external crossing edges. & How to choose a source and orientation that reduce the damage before component repair; link-safety filtering and runtime discovery. \\
This paper & Selects a damage-reducing root--orientation pair, filters failed EJ links, repairs residual components with the $c-1$ optimum, and characterizes recovery by healthy-graph connectivity. & Global optimization over all possible roots when the ranked/capped selector is used; tighter high-probability connectivity bounds for $G'$ under non-uniform or adversarial fault distributions. \\
\bottomrule
\end{tabular}
\end{adjustbox}
\end{table}

\subsection{Contributions}
The contributions are as follows.
\begin{itemize}
    \item[C1.] We introduce the selected triple $(r,\theta,\Kcomp_{r,\theta})$ and prove that hybrid recovery succeeds if and only if the healthy EJ graph $G'$ is connected---a necessary and sufficient characterization that subsumes node, link, mixed, and runtime faults.
    \item[C2.] We prove the EJ link-exclusion test: whether a selected broadcast tree uses a given failed link reduces to two coordinate parent comparisons, costing $O(|\Fe|)$ per candidate without constructing the full tree.
    \item[C3.] We prove edge-minimum component repair: when $\Kcomp_{r,\theta}$ is connected, exactly $c-1$ crossing edges are necessary and sufficient, and deterministic single-link repair follows because no EJ link is a bridge.
    \item[C4.] We prove a repaired-depth bound $D_{r,\theta}\le 2t+1$ under shallowest-layer entry selection, converting the near-diameter depth observed in deterministic regimes into a theorem.
    \item[C5.] We validate the framework over 260,000 trials at five network scales with a global-BFS rebuild audit, a cap-sensitivity audit, and a near-miss audit, demonstrating 100\% recovery and quantifying the structural repair-edge advantage over fixed-source repair.
\end{itemize}

\subsection{Scope}
The deterministic guarantees in this paper are deliberately limited but explicit. One/two-node boundary re-rooting is guaranteed. Single-link repair is guaranteed after a root--orientation pair is selected. The $c-1$ repair theorem is exact whenever the selected component graph is connected. For higher-order node/link faults, success is governed by healthy-graph connectivity: if the healthy EJ graph remains connected, the component graph of every selected pruned tree is connected and the hybrid repair succeeds; if the healthy graph is disconnected, no broadcast tree over all healthy vertices can exist. The ranked/capped part of the implementation affects which successful root--orientation pair is chosen and therefore the observed repair-edge count and depth, not the existence of a successful repair under healthy-graph connectivity, because the original healthy source and its orientations are pinned as fallback candidates.

\section{Related Work}
Algebraic interconnection networks use group, ring, or Cayley-graph structure to obtain symmetry, compact routing descriptions, and low-degree regularity. General treatments of interconnection networks, collective communication, and routing appear in standard texts on parallel algorithms and network design~\cite{Leighton1992,DallyTowles2004,Duato2003,Grama2003}. Dense EJ networks belong to this family: their degree-six adjacency comes from the six unit directions of the triangular/EJ lattice, and the finite quotient gives a compact hexagonal topology with small diameter~\cite{FlahiveBose2010,MartinezEJ2008}. Hexagonal mesh and HARTS-style systems motivated early addressing, routing, and reliable broadcast algorithms in six-neighbor networks~\cite{ChenShinKandlur1990,KandlurShin1991}, while later EJ formulations supplied the algebraic model used here~\cite{AlbaderBoseFlahive2012}.

Dense EJ networks belong to the broader family of algebraic quotient and circulant interconnection networks, where vertex transitivity and compact generator sets support concise routing and broadcasting rules~\cite{Akers1989,FlahiveBose2010,MartinezEJ2008}. The present paper does not rely on results from other quotient-lattice families; it is self-contained for the degree-six EJ topology. The key recovery mechanisms--the distance-$t$ boundary re-rooting lemma, the two-comparison link-exclusion test, and the six-neighbor local obstruction model--all use the triangular EJ coordinate system and are proved directly in this manuscript.

Fault-tolerant broadcasting has also been studied in broader distributed and parallel-computer settings. Classical work considered time bounds and impossibility tradeoffs for fault-tolerant broadcast in distributed networks~\cite{PelegSchafer1989,HadzilacosToueg1994}. Hypercube fault-tolerant routing and broadcasting developed another line of results in which a structured topology is preserved while faulty nodes or links are bypassed~\cite{LeeHayes1988}. Independent or completely independent spanning trees provide tree diversity by precomputing multiple delivery structures, and surveys summarize this protection-oriented approach~\cite{ChengSurvey2023}. Network-on-Chip fault tolerance similarly includes adaptive routing, region avoidance, spare-path use, virtual-channel based rerouting, and topology-agnostic deadlock-free schemes~\cite{Kliazovich2013,PasrichaDutt2008,FlichBertozzi2010,Stroobant2018}.

The present work occupies a different point in this design space. It does not precompute several complete trees, duplicate every broadcast along multiple paths, or optimize arbitrary packet traffic under background load. Instead, it starts from a deterministic EJ diameter-level broadcast tree after a small fault set is known, selects a root--orientation pair to reduce the damaged component structure, and adds the minimum number of external component-crossing edges for the selected pruned tree. The external global-BFS rebuild baseline is intentionally included to make this distinction clear: BFS rebuild can restore reachability with small depth, but it may replace many parent pointers and therefore does not optimize local component-preserving repair state.

To our knowledge, no prior published method addresses joint root--orientation selection for component-minimizing repair of EJ broadcast trees under mixed node and link faults. The closest EJ-specific foundations are node-only boundary re-rooting and multi-orientation component repair. Those works solve disjoint subproblems; the present manuscript is self-contained and restates the results it needs, using them as ingredients in the coupled triple $(r,\theta,\Kcomp_{r,\theta})$. The irreducible contribution of this paper is therefore not a new component-repair certificate alone nor a new re-rooting guarantee alone, but the proof that selecting the coupled triple $(r,\theta,\Kcomp_{r,\theta})$ gives an exact recovery condition, together with the link-exclusion test that makes this selection computationally tractable. Boundary re-rooting by itself does not address link faults or characterize when repair fails; fixed-source component repair by itself does not prove that the component graph of every selected candidate is connected whenever the healthy graph is connected. The healthy-graph connectivity theorem closes both gaps simultaneously. The global-BFS rebuild baseline included in the experiments is therefore not a competing EJ broadcast-repair method in the same objective class. It is an intentionally strong external reachability baseline that shows the cost of discarding healthy tree structure entirely, making the repair-edge reduction advantage of the hybrid method interpretable.

Reliable broadcast algorithms for HARTS-style hexagonal meshes address faulty processors by storing alternate paths and recovery rules in the network~\cite{KandlurShin1991}. Such protection-oriented methods improve robustness but maintain redundancy before faults are known. In contrast, the hybrid method stores a constant family of EJ orientation rules and inserts only $O(c-1)$ exceptional forwarding rules after the fault set is observed. For the single-fault case, boundary re-rooting or single-link repair requires zero or one exceptional rule, whereas redundant path methods reserve alternate state regardless of whether it is needed.

Independent spanning tree methods precompute multiple complete spanning trees of the fault-free graph, providing tree diversity before faults occur~\cite{ChengSurvey2023}. They are complementary but solve a different optimization problem: an IST is fixed before the fault set is known, whereas the hybrid method chooses a root--orientation pair after faults are observed. A fault placement that creates many components in one precomputed tree may create fewer components in a differently rooted EJ orientation. The hybrid method exploits this post-fault selection freedom and then proves $c-1$ repair optimality for the selected pruned tree.

\section{Dense EJ Network Model}
Let
\begin{equation}
\omega=\frac{-1+i\sqrt{3}}{2}.
\end{equation}
An EJ integer is written as $x+y\omega$ and represented by the axial coordinate pair $(x,y)\in\mathbb{Z}^2$. The six unit directions are
\begin{equation}
\D=\{(1,0),(0,1),(1,-1),(-1,0),(0,-1),(-1,1)\}.
\end{equation}
The EJ distance from the origin is
\begin{equation}
\rhoEJ(x,y)=\max\{|x|,|y|,|x+y|\}.
\end{equation}
The radius-$t$ hexagonal ball is
\begin{equation}
\Ht=\{(x,y): |x|\le t,\ |y|\le t,\ |x+y|\le t\}.
\end{equation}
It contains
\begin{equation}
N=3t^2+3t+1
\end{equation}
vertices. The dense EJ network considered in this paper is generated by
\begin{equation}
\alpha=(t+1)+t\omega,
\end{equation}
and is the quotient of the EJ lattice modulo $\alpha$. The integer label associated with $(x,y)$ is
\begin{equation}
\phi(x+y\omega)\equiv tx-(t+1)y\pmod N.
\end{equation}
Thus unit axial moves correspond to the six degree-six circulant jumps
\begin{equation}
\pm t,\quad \pm(t+1),\quad \pm(2t+1)
\end{equation}
under the chosen convention.

For a node $a$ and offset $z\in\Ht$, the translated node is denoted by $a+z$ in the quotient. The relative coordinate of node $v$ with respect to root $r$ is
\begin{equation}
\Delta_r(v)=\can(v-r)\in\Ht,
\end{equation}
where $\can$ returns the canonical axial representative in the radius-$t$ hexagon.

\begin{figure}[H]
\centering
\includegraphics[width=0.6\linewidth]{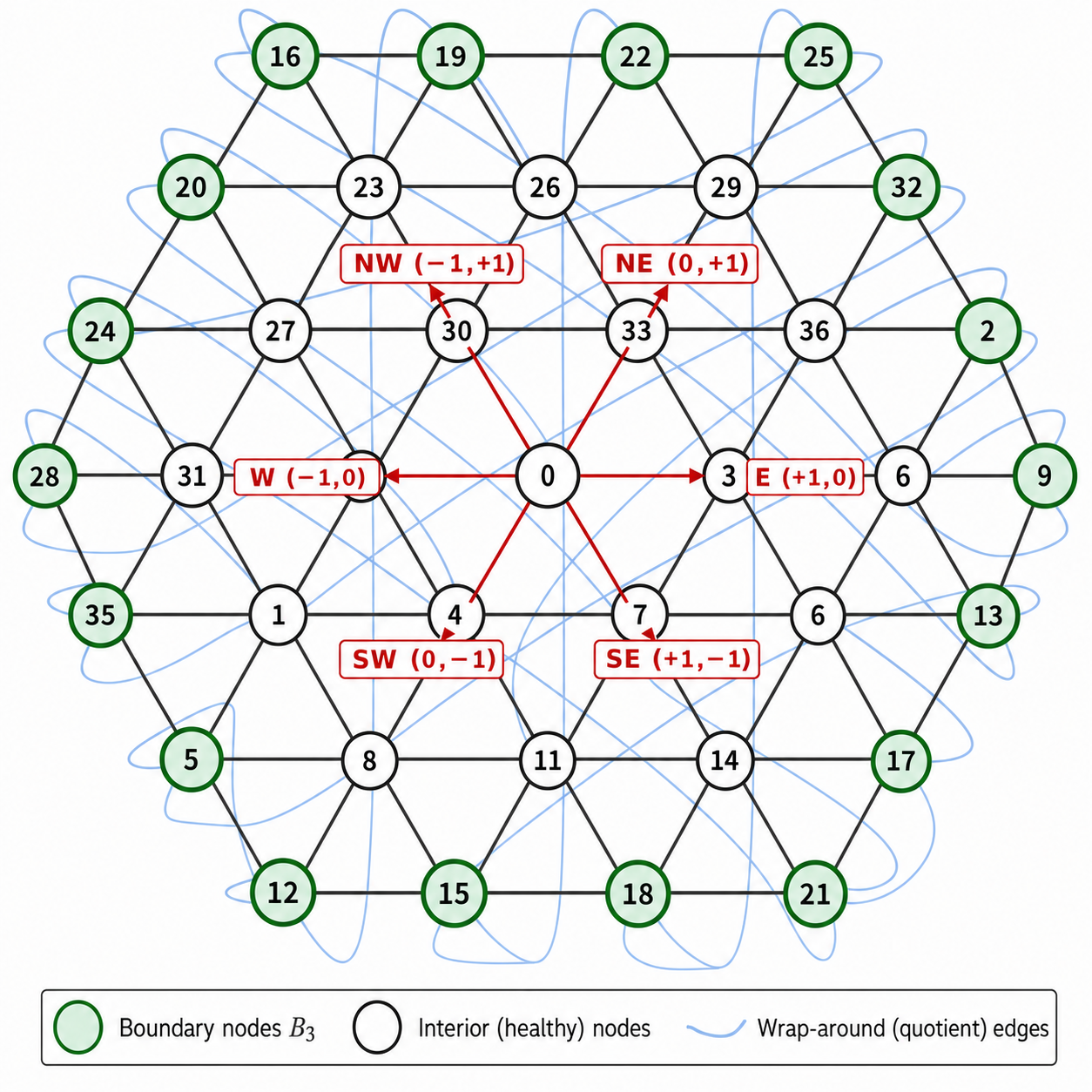}
\caption{Dense EJ network for $t=3$ ($n=4$, $N=37$) represented as an axial hexagon. Boundary nodes satisfy $\rho(x,y)=3$ and are highlighted. Integer labels use $\phi(x+y\omega)=3x-4y \pmod{37}$, and the six annotated directions correspond to the six unit EJ moves.}
\label{fig:ej-hexagon}
\end{figure}

\section{Rooted EJ Broadcast Orientations}
A coordinate-reduction orientation chooses one parent for every non-root vertex by selecting an inward neighbor whose EJ layer is one smaller.

\begin{definition}[EJ coordinate-reduction parent rule]
Let $\theta=(d_1,d_2,\ldots,d_6)$ be an ordered list of the six EJ unit directions. For a nonzero coordinate $z\in\Ht$, define
\begin{equation}
 p_\theta(z)=z+d_j,
\end{equation}
where $d_j$ is the first direction in $\theta$ satisfying
\begin{equation}
\rhoEJ(z+d_j)=\rhoEJ(z)-1.
\end{equation}
\end{definition}

For a selected root $r$, the tree $\T_{r,\theta}$ contains the directed parent edge
\begin{equation}
 r+p_\theta(\Delta_r(v))\longrightarrow v
\end{equation}
for every non-root node $v$ before faults are applied.

\begin{lemma}[Fault-free EJ orientation tree]
For any root $r$ and any orientation $\theta$ that assigns to every non-root coordinate an inward parent, $\T_{r,\theta}$ is a spanning tree of the dense EJ network rooted at $r$ with depth at most $t$.
\end{lemma}

\begin{proof}
Every non-root vertex has exactly one parent. Along each parent edge, the layer $\rhoEJ(\Delta_r(v))$ decreases by one. Hence directed cycles are impossible and repeated parent application reaches the unique layer-zero coordinate, namely the root. Since every canonical coordinate has layer at most $t$, the maximum root-to-node depth is at most $t$.
\end{proof}

\subsection{Orientation Family}
The hybrid framework can use either one fixed EJ broadcast orientation or a constant-size family. For compatibility with EJ-MOEM, we use the 15-orientation family. Let the six directions be indexed cyclically by $\delta_0,\ldots,\delta_5$. Define
\begin{align}
C_i&=(\delta_i,\delta_{i+1},\delta_{i+2},\delta_{i+3},\delta_{i+4},\delta_{i+5}),\quad 0\le i<6,\\
R_i&=(\delta_i,\delta_{i-1},\delta_{i-2},\delta_{i-3},\delta_{i-4},\delta_{i-5}),\quad 0\le i<6,\\
A_i&=(\delta_i,\delta_{i+3},\delta_{i+1},\delta_{i+4},\delta_{i+2},\delta_{i+5}),\quad 0\le i<3,
\end{align}
with indices modulo six. The selected family affects repair quality and depth, not the graph-theoretic optimality theorem for a chosen tree.

\section{Fault Model and Recovery Objective}
Let $\Fv\subseteq V$ be the faulty-node set and $\Fe\subseteq E$ the failed-link set. The original source is assumed healthy. A recovered broadcast must reach every node in $V\setminus\Fv$ without using failed nodes or failed links.

For a selected root--orientation pair $(r,\theta)$, pruning removes faulty nodes, tree edges incident to faulty nodes, and failed tree links:
\begin{equation}
\T^-_{r,\theta}=\T_{r,\theta}-\Fv-\Fe.
\end{equation}
Let
\begin{equation}
\Cset_{r,\theta}=\{C_1,C_2,\ldots,C_c\}
\end{equation}
be the connected components of $\T^-_{r,\theta}$ on healthy vertices.

\begin{definition}[Healthy component graph]
The healthy component graph $\Kcomp_{r,\theta}$ has one vertex for each component $C_i$. Two component vertices $C_i$ and $C_j$ are adjacent if there exists a healthy EJ graph edge $uv\in E\setminus\Fe$ with $u\in C_i$, $v\in C_j$, and $u,v\notin\Fv$.
\end{definition}

For a successful candidate, let $\rho_{r,\theta}$ be the number of external component-crossing repair edges, $D_{r,\theta}$ the repaired depth, and $M_{r,\theta}$ the number of inspected root/orientation candidates. The hybrid selector used in this paper orders valid candidates lexicographically as
\begin{equation}
\min (\rho_{r,\theta},D_{r,\theta},M_{r,\theta}).
\end{equation}
Thus the primary metric is exceptional forwarding state, while repaired depth is reported separately as a latency-related metric. This paper does not evaluate a separate weighted latency objective; avoiding an unevaluated objective prevents confusion between the proven edge-minimum repair certificate and possible future latency-dominant engineering choices.

An \emph{avoid-only} candidate is one for which $\lambda(r,\theta,\Fe)=0$ and every node in $\Fv$ that can be neutralized by re-rooting satisfies $\dist(r,f)=t$. Equivalently, the selected tree neither uses any failed link nor has a neutralized faulty node at an internal forwarding position. In that case zero component-crossing repair edges are needed, and the candidate succeeds without invoking the component-repair phase.

\section{EJ Re-Rooting Preliminaries}
The graph-distance-$t$ boundary of the origin is
\begin{equation}
\Bt=\{z\in \Ht:\rhoEJ(z)=t\}.
\end{equation}
The distance-$t$ boundary of a node $a$ is $a+\Bt$.

\begin{lemma}[EJ boundary-difference coverage]
For the dense EJ network of diameter $t$,
\begin{equation}
\Bt-\Bt=V.
\end{equation}
Equivalently, for every displacement $A\in V$, there exist $U,V\in\Bt$ such that $A=U-V$ in the EJ quotient.
\end{lemma}

\begin{proof}
Every node has a canonical representative $A=(x,y)\in\Ht$. It is enough to show that every point of $\Ht$ is a difference of two boundary points. In the sector $x\ge0$, $y\ge0$, and $x+y\le t$, write $A=r+s\omega$ with $r,s\ge0$ and $r+s\le t$. Define
\begin{equation}
U=(r-t)+t\omega,
\qquad
V=-t+(t-s)\omega.
\end{equation}
Then $U-V=r+s\omega=A$. Moreover,
\begin{align}
\rhoEJ(r-t,t)&=\max\{t-r,t,r\}=t,\\
\rhoEJ(-t,t-s)&=\max\{t,t-s,s\}=t,
\end{align}
so $U,V\in\Bt$. The remaining five sectors follow by multiplication by EJ units, which preserve adjacency, distance, and the boundary. Therefore every displacement in $V$ is in $\Bt-\Bt$.
\end{proof}

\begin{theorem}[One/two-node EJ re-rooting]
Let $\Fv$ be a node-fault set with $|\Fv|\le2$. There exists a root $r$ such that
\begin{equation}
\dist(r,f)=t\quad \text{for every }f\in\Fv.
\end{equation}
\end{theorem}

\begin{proof}
For one fault $f$, choose any $r\in f+\Bt$. For two faults $f_1,f_2$, set $A=f_2-f_1$. By boundary-difference coverage, choose $U,V\in\Bt$ with $A=U-V$. Let $r=f_1+U$. Then $r-f_1=U$, so $\dist(r,f_1)=t$. Also $r-f_2=f_1+U-f_2=U-A=V$, so $\dist(r,f_2)=t$.
\end{proof}

\begin{corollary}[Zero-repair node-only cases]
If $|\Fv|\le2$, $\Fe=\emptyset$, and the selected EJ broadcast tree treats every distance-$t$ node as a leaf, then a root satisfying the preceding theorem makes every faulty node a leaf. Removing those faulty leaves does not disconnect the healthy tree, so zero external repair edges are required.
\end{corollary}

\begin{figure}[H]
\centering
\includegraphics[width=0.75\linewidth]{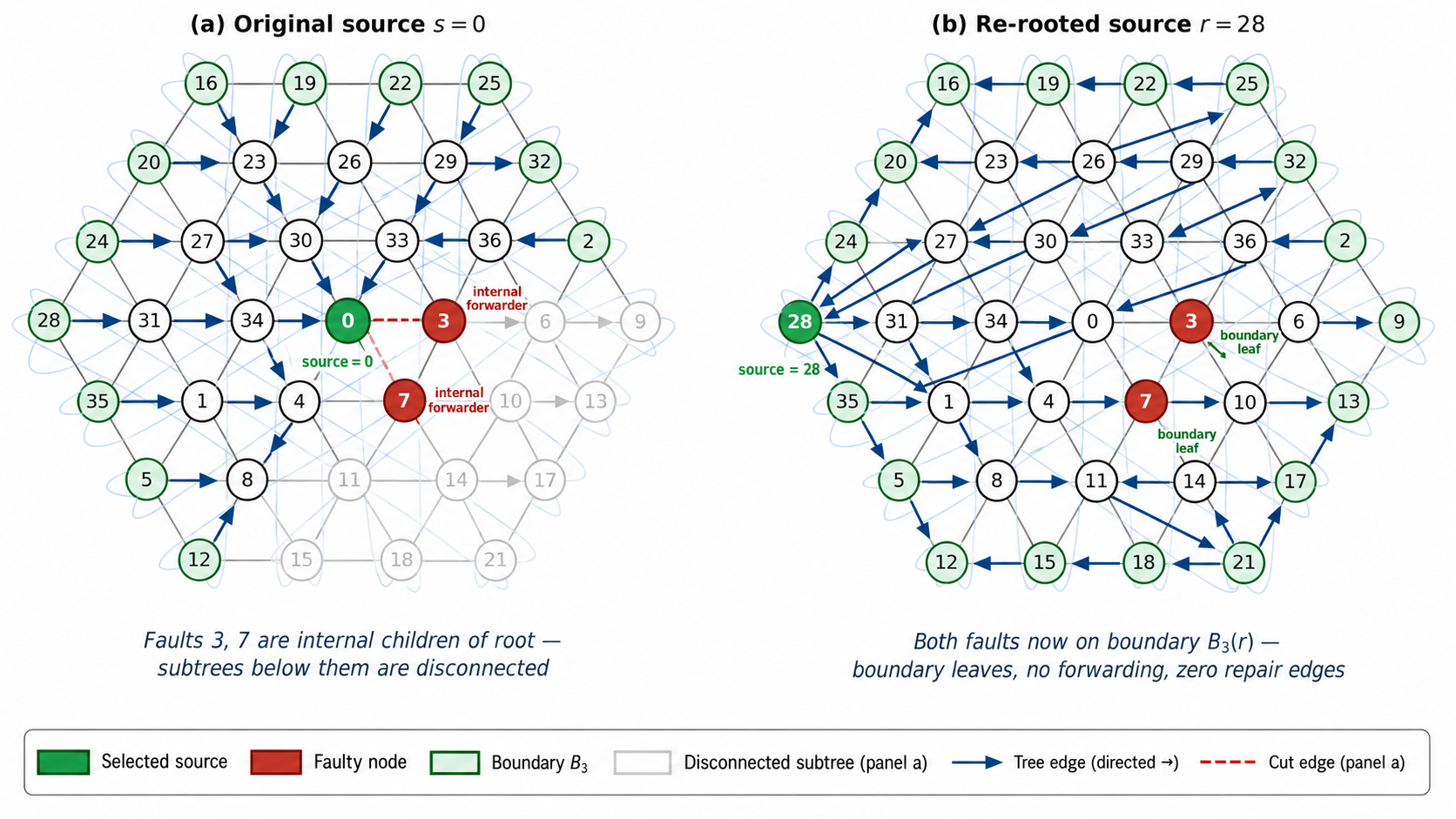}
\caption{Two-fault re-rooting on $H_4$. Under the original source, the two faults can lie on internal forwarding branches. After re-rooting to $r=28$, both faulty nodes lie on the distance-$3$ boundary of the selected root and become leaves, so no component repair is needed in the node-only case.}
\label{fig:rerooting-before-after}
\end{figure}

\begin{remark}[Three-fault limitation]
The one/two-node guarantee does not extend to arbitrary triples. A zero-repair three-node re-rooting exists precisely when
\begin{equation}
\bigcap_{f\in\Fv}\{r:\dist(r,f)=t\}\ne\emptyset.
\end{equation}
There are triples in $H_4$ with empty triple boundary intersection. Thus all higher-order node-fault claims in this hybrid framework must be stated as conditional or empirical unless an additional fault-class restriction is imposed.
\end{remark}

\begin{example}[Concrete $H_4$ three-fault obstruction]
Let $t=3$ and $N=37$. In the integer-label representation with jumps $3,4,7$, take
\begin{equation}
F_1=0,\qquad F_2=5,\qquad F_3=14.
\end{equation}
The distance-three boundary of label $0$ is
\begin{multline}
B_3(0)=\{2,5,9,12,13,15,16,17,18,19,20,21,\\
22,24,25,28,32,35\}.
\end{multline}
By vertex transitivity, $B_3(F_i)=F_i+B_3(0)$ modulo $37$. Direct intersection gives
\begin{equation}
B_3(F_1)\cap B_3(F_2)\cap B_3(F_3)=\emptyset.
\end{equation}
Therefore no re-rooted source is at distance $t$ from all three faults. This example is used only to mark the deterministic boundary of the node-only re-rooting theorem; the hybrid method may still recover many three-fault instances after component repair.
\end{example}

\section{EJ Link-Failure Exclusion}
Node re-rooting uses distance-to-boundary constraints. Link failures require a different condition: a failed link matters only if the selected broadcast tree uses that link as a parent--child edge.

\begin{definition}[Failed-tree-link count]
For a root--orientation pair $(r,\theta)$ and failed-link set $\Fe$, define
\begin{equation}
\lambda(r,\theta,\Fe)=|E(\T_{r,\theta})\cap \Fe|.
\end{equation}
A candidate is link-safe if $\lambda(r,\theta,\Fe)=0$.
\end{definition}

\begin{proposition}[EJ link-exclusion test]
Let $e=\{u,v\}$ be a failed EJ link. For a selected root $r$ and orientation $\theta$, write $\Delta_r(w)$ for the canonical coordinate of $w$ relative to $r$. Then $\T_{r,\theta}$ uses $e$ if and only if
\begin{equation}
\Delta_r(u)=p_\theta(\Delta_r(v))
\quad\text{or}\quad
\Delta_r(v)=p_\theta(\Delta_r(u)).
\end{equation}
Thus testing whether a candidate uses one failed link requires only constant-size coordinate comparisons after canonical reduction.
\end{proposition}

\begin{proof}
By definition of the rooted orientation tree, the parent of a non-root node $w$ is the quotient node with relative coordinate $p_\theta(\Delta_r(w))$. Therefore $u$ is the parent of $v$ if and only if $\Delta_r(u)=p_\theta(\Delta_r(v))$. Similarly, $v$ is the parent of $u$ if and only if $\Delta_r(v)=p_\theta(\Delta_r(u))$. Since an undirected failed link is used by the directed broadcast tree exactly in one of these two parent--child cases, the equivalence follows.
\end{proof}

\begin{corollary}[Link-safe roots and orientations]
For a failed-link set $\Fe$, a candidate $(r,\theta)$ is link-safe if and only if the two parent equalities in the link-exclusion test are false for every $e\in\Fe$. Testing one candidate costs $O(|\Fe|)$ coordinate comparisons and does not require constructing the full tree.
\end{corollary}

\begin{example}[Link-exclusion test on $H_4$]
Let $t=3$, $r=0$, and use the first cyclic orientation. The node with label $3$ has coordinate $(1,0)$ and parent label $0$, so the failed link $\{0,3\}$ satisfies the endpoint-parent test and is a failed tree link. By contrast, labels $28$ and $35$ correspond to adjacent boundary coordinates $(-3,0)$ and $(-2,-1)$, but neither endpoint is the selected parent of the other under this orientation; the link is an EJ graph edge but not a tree edge. Thus the candidate using $r=0$ is unsafe for $\{0,3\}$ but safe for $\{28,35\}$. The test uses only the two parent comparisons in Proposition~1 and therefore costs $O(|F_E|)$ for a fixed candidate.
\end{example}

\section{Edge-Minimum Component Repair}
The following theorem is independent of the special EJ geometry once the tree and component graph are fixed. It is included here because it is the central repair certificate of the hybrid framework.

\begin{lemma}[Component lower bound]
If $\T^-_{r,\theta}$ has $c$ healthy components, any non-redundant repaired broadcast tree that preserves those internal components must use at least $c-1$ external component-crossing repair edges.
\end{lemma}

\begin{proof}
Contract each healthy component into one supernode. A repaired broadcast tree spanning all healthy nodes must connect the $c$ supernodes. Any connected graph on $c$ vertices has at least $c-1$ edges. Each such edge corresponds to a component-crossing repair edge in the original network.
\end{proof}

\begin{theorem}[Edge-minimum repair for a selected EJ tree]
If the healthy component graph $\Kcomp_{r,\theta}$ is connected, then $\T^-_{r,\theta}$ can be repaired into a non-redundant broadcast tree over all healthy nodes using exactly $c-1$ external component-crossing repair edges. This number is minimum.
\end{theorem}

\begin{proof}
Since $\Kcomp_{r,\theta}$ is connected, choose any spanning tree of $\Kcomp_{r,\theta}$ rooted at the component containing $r$. For each component-level spanning-tree edge, add one corresponding healthy EJ crossing edge between the two original components. Inside each component, the original structure is a tree because it is a subgraph of $\T_{r,\theta}$. Contracting each component maps the repaired structure to the selected component-level spanning tree, so no cycle is introduced. The final graph is connected, acyclic, spans all healthy vertices, excludes faulty nodes and failed links, and can be rooted at $r$. It uses exactly $c-1$ crossing edges, which is minimum by the component lower bound.
\end{proof}

\begin{figure}[H]
\centering
\includegraphics[width=0.75\linewidth]{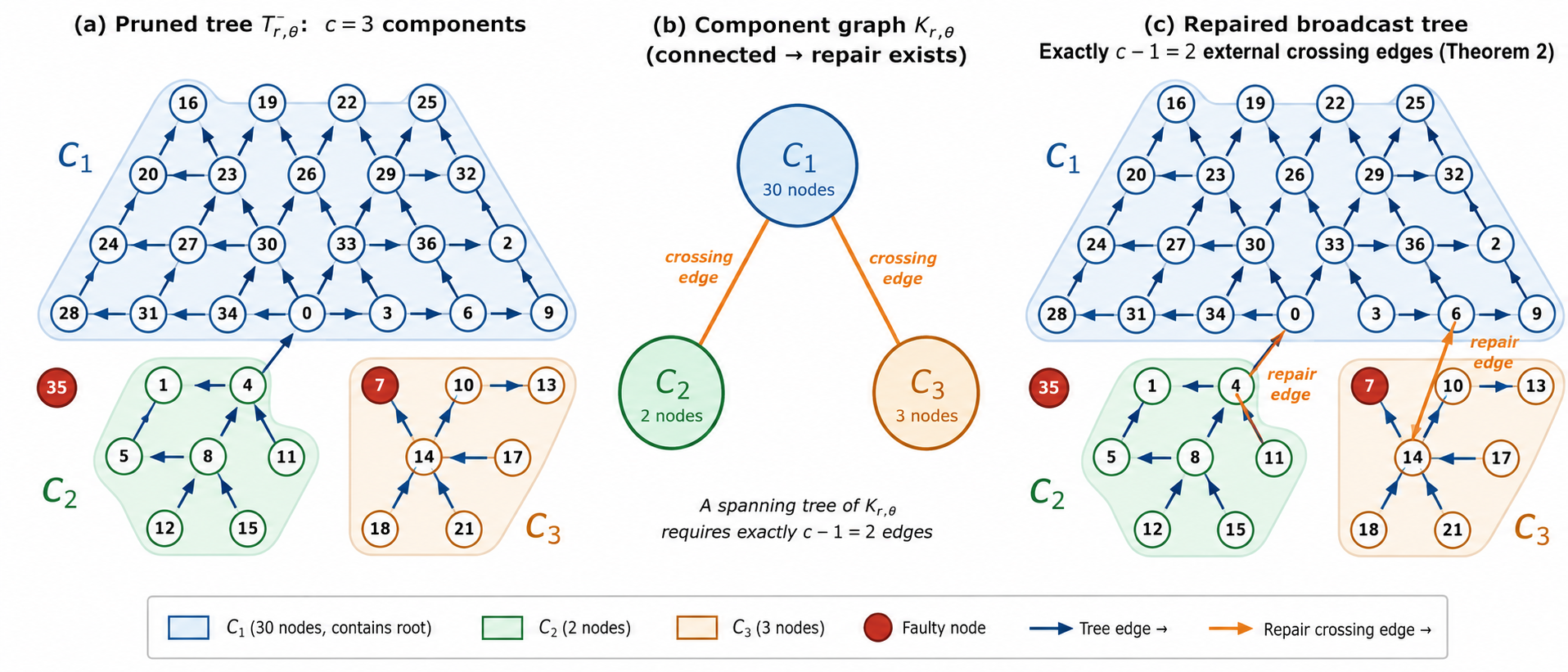}
\caption{Component contraction and edge-minimum repair. After pruning, the healthy forest has three components. The connected component graph admits a component-level spanning tree. Mapping that spanning tree back to EJ crossing edges reconnects the original components using exactly $c-1=2$ external repair edges.}
\label{fig:component-repair}
\end{figure}

\begin{remark}[Repair count versus repaired depth]
The repair-edge count $c-1$ counts exceptional component-crossing forwarding rules. The repaired depth is a separate metric because a repair edge may attach to a component at a vertex far from that component's farthest descendant. Therefore the hybrid framework reports both repair-edge count and repaired depth.
\end{remark}

\begin{lemma}[Repaired-depth accounting bound]
Let $\T^-_{r,\theta}$ have healthy components $C_1,\ldots,C_c$, with $C_1$ containing the selected root $r$. For a non-root component $C_i$, define
\begin{equation}
\ell(C_i)=\min_{v\in C_i} d_{\T_{r,\theta}}(r,v),\qquad
\operatorname{span}(C_i)=\max_{v\in C_i} d_{\T_{r,\theta}}(r,v)-\ell(C_i).
\end{equation}
If the component-level repair attaches $C_i$ through a crossing edge whose already-repaired endpoint has repaired depth $d_{\rm attach}$ and enters $C_i$ at a vertex $b_i$, then every vertex of $C_i$ is reached at repaired depth at most
\begin{equation}
 d_{\rm attach}+1+\ecc_{C_i}(b_i).
\end{equation}
If $b_i$ is a shallowest-layer entry of $C_i$, then $\ecc_{C_i}(b_i)\le \operatorname{span}(C_i)$; in general $\ecc_{C_i}(b_i)\le 2\operatorname{span}(C_i)$. Since every component is a subtree of an EJ coordinate-reduction tree of depth at most $t$, $\operatorname{span}(C_i)\le t$. Consequently, whenever the selected component repair attaches every non-root component from an endpoint of repaired depth at most $t$, the repaired depth satisfies
\begin{equation}
D_{r,\theta}\le 3t+1,
\end{equation}
and it improves to $D_{r,\theta}\le 2t+1$ when each component is entered through a shallowest-layer entry. If the entering crossing edge is counted as the first edge of the component suffix, this is reported as the cleaner $2t$ accounting bound used in the deterministic suffix cases.
\end{lemma}

\begin{proof}
After a crossing edge is added from an already-repaired endpoint to an entry vertex $b_i$ of $C_i$, every vertex of $C_i$ is reached by following the repaired path to the crossing endpoint, traversing the crossing edge, and then following tree edges inside $C_i$. The maximum additional distance inside the component is $\ecc_{C_i}(b_i)$. If $b_i$ is a shallowest-layer entry, the farthest descendant lies at most the layer span away along the coordinate-reduction suffix. The shallowest-layer entry assumption is enforceable in Algorithm~1 at the component-edge selection step by selecting, among all healthy crossing edges between a parent component and $C_i$, an edge whose endpoint in $C_i$ has minimum layer index. This choice is computable during the component-edge scan, does not increase the repair-edge count, and is therefore compatible with the $c-1$ optimality of the edge-minimum repair theorem.

Every component $C_i$ is a connected subtree of $\T_{r,\theta}$, which is a rooted tree whose vertices are stratified by the EJ layer $\rho(\Delta_r(v))$. Within $C_i$, let $\ell_{\min}$ and $\ell_{\max}$ be the minimum and maximum layer indices, so $\operatorname{span}(C_i)=\ell_{\max}-\ell_{\min}$. Because $C_i$ is a connected component of a pruned rooted tree, a minimum-layer vertex $v^*$ is the local root of $C_i$ and an ancestor of every other vertex in $C_i$; otherwise the parent path from a vertex of $C_i$ to the global root would leave and re-enter $C_i$, contradicting component maximality after pruning. For an arbitrary entry vertex $b_i\in C_i$ and any vertex $w\in C_i$, the unique tree path inside $C_i$ is contained in the union of the two ancestor paths through $v^*$, so
\begin{equation}
 d_{C_i}(b_i,w)\le d_{C_i}(b_i,v^*)+d_{C_i}(v^*,w).
\end{equation}
The path from $b_i$ to $v^*$ follows parent edges through layers between $\ell_{\min}$ and $\ell_{\max}$, so its length is at most $\rho(\Delta_r(b_i))-\ell_{\min}\le \operatorname{span}(C_i)$. Similarly, the ancestor path from $v^*$ to $w$ has length at most $\ell_{\max}-\ell_{\min}=\operatorname{span}(C_i)$. Hence $\ecc_{C_i}(b_i)\le 2\operatorname{span}(C_i)$. Since every layer lies between $0$ and $t$, $\operatorname{span}(C_i)\le t$. Substitution proves the bounds.
\end{proof}

\begin{remark}[Why deterministic rows have depth $t$]
In the deterministic one/two-node rows, the selected re-rooted source places each faulty node on the distance-$t$ boundary. Boundary nodes are leaves of the EJ broadcast tree, so deleting them does not create internal detached components. In the single-link rows, either the failed link is avoided or the bypass reconnects the detached suffix at the same layer where the cut occurred. Thus the deepest reachable healthy boundary layer remains $t$, and the repaired depth is exactly $t$ in the deterministic rows of Table~\ref{tab:repair-reduction-v4}. Higher-order rows can have larger delivery tails because several faults may detach internal components whose best entry is not at the original shallowest layer.
\end{remark}

\begin{corollary}[Deterministic-regime exact depth]
In the one/two-node boundary re-rooting regime $(|\Fv|\le 2,\Fe=\emptyset)$, the selected re-rooted tree has repaired depth $D_{r,\theta}=t$ whenever at least one healthy boundary vertex remains. In the single-link regime $(\Fv=\emptyset, |\Fe|=1)$ with a link-safe candidate, $D_{r,\theta}=t$.
\end{corollary}

\begin{proof}
In the one/two-node case, the selected root places every faulty node at layer $t$. These nodes are leaves and their removal does not detach any internal component, so the pruned tree keeps the same layer structure as $\T_{r,\theta}$ on the healthy vertices. The boundary $B_t$ has $|B_t|=6t$ vertices, and at most two of them are faulty. For every $t\ge1$, at least $6t-2\ge4$ boundary vertices remain healthy. Hence a healthy layer-$t$ vertex remains in the pruned tree, so the maximum repaired depth is exactly $t$. In the link-safe single-link case, $\lambda=0$ and the pruned tree is identical to $\T_{r,\theta}$ restricted to healthy nodes; the depth is again $t$.
\end{proof}

\begin{example}[Repaired-depth accounting]
Suppose a higher-order placement in $H_5$ ($t=4$) detaches a component $C_i$ whose shallowest layer is $2$ and whose deepest layer is $4$, so $\operatorname{span}(C_i)=2$. If the component is attached from a repaired endpoint at depth $3$, an arbitrary entry may reach the farthest vertex at depth up to $3+1+2=6$, giving overhead $2$ above the diameter. If a shallower crossing edge enters $C_i$ at its minimum-layer entry, the same component is reached with a smaller internal suffix, e.g., depth $3+1+1=5$. This is why Algorithm~1 records the shallowest entry among crossing edges when a depth-bounded certificate is requested: it does not change the $c-1$ repair count, but it reduces delivery-tail depth.
\end{example}

\section{Deterministic Single-Link Repair}
A single failed link has a deterministic repair guarantee in dense EJ networks.

\begin{lemma}[No EJ link is a bridge]
Every EJ graph edge lies on a triangle in the dense EJ quotient. Therefore deleting one link cannot disconnect the dense EJ network.
\end{lemma}

\begin{proof}
Let $e=(u,u+d)$ where $d$ is one of the six EJ unit directions. Choose a unit direction $h$ adjacent to $d$ in the triangular lattice, so that $h-d$ is also a unit direction. Then
\begin{equation}
 u\longrightarrow u+d\longrightarrow u+h\longrightarrow u
\end{equation}
uses the unit directions $d$, $h-d$, and $-h$, respectively. Thus $e$ lies on a 3-cycle, possibly represented through quotient wraparound at the boundary. Hence no EJ link is a bridge.
\end{proof}

\begin{lemma}[Healthy graph implies component graph]
If the healthy EJ graph $H_t-\Fv-\Fe$ is connected, then the component graph $\Kcomp_{r,\theta}$ is connected for every pruned tree $\T^-_{r,\theta}$ over the healthy vertices.
\end{lemma}

\begin{proof}
Take any two tree components. Since the healthy graph is connected, there is a healthy graph path between them. Whenever this path moves from one tree component to another, it uses a healthy crossing edge. Hence the path induces a path in the component graph.
\end{proof}

\begin{theorem}[Healthy-graph connectivity is necessary and sufficient]
Let $\Fv$ and $\Fe$ be arbitrary fault sets with the original source healthy, and let
\begin{equation}
G'=H_t-\Fv-\Fe
\end{equation}
be the healthy EJ graph. With the mandatory healthy-source fallback used by Algorithm~1, a repaired broadcast tree over $V\setminus\Fv$ exists if and only if $G'$ is connected. In particular, when $G'$ is connected, the component graph $\Kcomp_{r,\theta}$ is connected for every retained root--orientation pair $(r,\theta)$, and the repair phase succeeds with exactly $c-1$ external component-crossing edges for that selected pruned tree.
\end{theorem}

\begin{proof}
For sufficiency, suppose $G'$ is connected. By the preceding lemma, for every retained root--orientation pair $(r,\theta)$, the component graph $\Kcomp_{r,\theta}$ of the pruned tree is connected. The edge-minimum repair theorem then adds exactly $c-1$ healthy crossing edges and returns a valid non-redundant broadcast tree over all healthy vertices. Algorithm~1 pins the original healthy source with all orientations in its mandatory fallback set, so at least one healthy retained candidate is evaluated even when the ranked top-$K$ list is capped; hence it succeeds. For necessity, suppose $G'$ is disconnected. Then two healthy vertices lie in different connected components of the healthy graph. No tree using only healthy vertices and healthy links can span both vertices, so no broadcast repair method can succeed.
\end{proof}

\begin{remark}[Candidate retention]
The ranked/capped selector never removes the original source fallback. Algorithm~1 evaluates the union of the top-$K$ ranked root--orientation pairs and the mandatory fallback set containing the original healthy source with all 15 orientations. Thus the cap can affect which high-ranked candidate wins the lexicographic score, but it cannot eliminate the existence of a successful retained candidate whenever $G'$ is connected.
\end{remark}

\begin{corollary}[High-order recovery condition]
The hybrid method fails to produce a repaired broadcast tree exactly when the healthy graph $G'=H_t-\Fv-\Fe$ is disconnected, or when an implementation deliberately evaluates no healthy candidate. Under adversarial fault placement, recovery is governed solely by healthy-graph connectivity: if the adversary's fault set leaves $G'$ connected, the hybrid repair succeeds for that placement; if the adversary disconnects $G'$, no broadcast method can succeed by the necessity direction of the healthy-graph connectivity theorem. The probability bound below therefore quantifies only the likelihood of random-placement local obstructions, not the existence of adversarial ones.
\end{corollary}

\begin{figure}[H]
\centering
\includegraphics[width=0.75\linewidth]{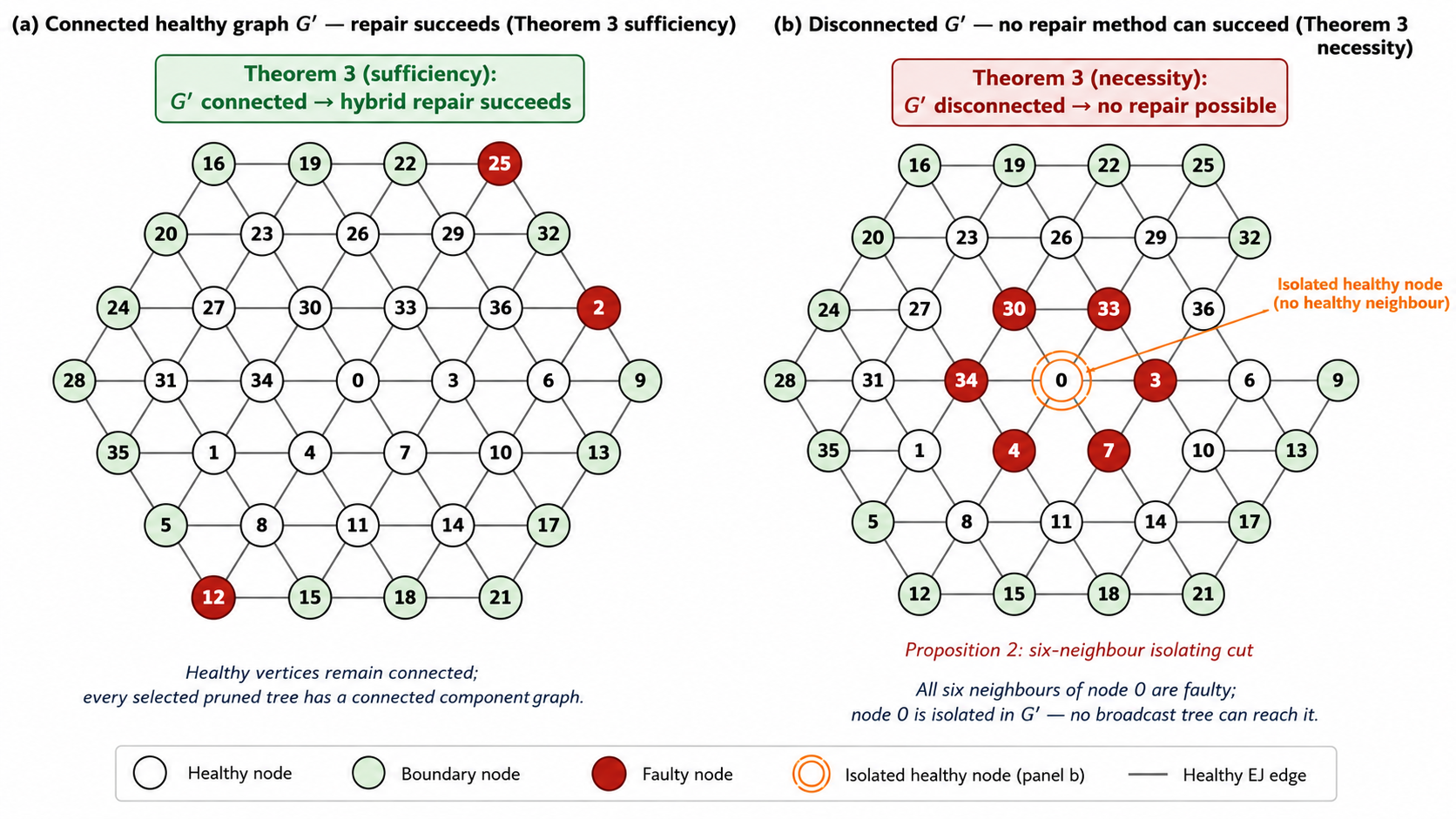}
\caption{Healthy-graph connectivity is the exact recovery condition. If the healthy graph $G'=H_t-F_V-F_E$ remains connected, every selected pruned tree has a connected component graph and hybrid repair succeeds. If a healthy node is isolated by six faulty neighbours, $G'$ is disconnected and no broadcast tree can reach all healthy nodes.}
\label{fig:healthy-connectivity}
\end{figure}

\begin{theorem}[Deterministic single-link repair]
Let $\Fv=\emptyset$ and let $\Fe=\{e\}$ contain one failed EJ link. For any selected healthy root $r$ and orientation $\theta$, the hybrid repair phase succeeds. If $e\notin E(\T_{r,\theta})$, zero repair edges are needed. If $e\in E(\T_{r,\theta})$, exactly one repair edge is necessary and sufficient.
\end{theorem}

\begin{proof}
If $e$ is not a tree edge, pruning does not change the tree and no repair is needed. If $e$ is a tree edge, deleting it splits $\T_{r,\theta}$ into exactly two healthy components. By the no-bridge lemma, deleting one EJ link does not disconnect the underlying healthy EJ graph. Therefore the two-component graph is connected. The edge-minimum repair theorem repairs it with $c-1=1$ crossing edge. The component lower bound shows that one crossing edge is necessary.
\end{proof}

\begin{corollary}[One-node/one-link bounded mixed case]
If one node fault is placed on the boundary by re-rooting and one failed link is present, then the selected tree either avoids the failed link or the failed link creates a two-component residual cut that is repaired by one edge, provided the healthy graph remains connected. In particular, the node fault itself does not create an internal forwarding cut under the boundary-leaf condition.
\end{corollary}

\section{Hybrid Re-Rooting and Repair Algorithm}
The hybrid algorithm evaluates a bounded candidate set of root--orientation pairs. For node faults, the re-rooting theorem provides exact candidates for $|\Fv|\le2$. For link-only or mixed cases, link-safety testing filters candidates; if no link-safe candidate is found, the algorithm selects the candidate that minimizes residual component repair.

\begin{algorithm}[H]
\caption{EJ Hybrid Re-Rooting-Assisted Edge-Minimum Repair}
\begin{algorithmic}[1]
\REQUIRE Dense EJ network $H_t$, original source $s$, node faults $\Fv$, link faults $\Fe$, orientation family $\Theta$.
\ENSURE Repaired non-redundant broadcast tree over $V\setminus\Fv$, or failure certificate.
\STATE Generate candidate roots $R$: exact boundary roots for $|\Fv|\le2$; otherwise leaf-score-ranked roots; if the original source $s$ is healthy, mark $(s,\theta)$ for every $\theta\in\Theta$ as mandatory fallback pairs.
\STATE Rank candidate root--orientation pairs cheaply by leaf score $L(r)$ and failed-tree-link count $\lambda(r,\theta,\Fe)$.
\STATE Set best $\leftarrow$ null.
\FOR{each pair in the union of the top-$K$ ranked pairs and the mandatory fallback pairs}
    \STATE Compute $\lambda(r,\theta,\Fe)$ using the EJ link-exclusion test.
    \STATE Prune $\T_{r,\theta}$ by deleting $\Fv$ and failed tree links in $\Fe$.
    \STATE Compute components $C_1,\ldots,C_c$ and the component graph $\Kcomp_{r,\theta}$.
    \IF{$\Kcomp_{r,\theta}$ is connected}
        \STATE Add a component-level spanning tree and map it to $c-1$ healthy EJ crossing edges, selecting for each newly attached component the crossing edge whose entry vertex in that component has minimum layer index.
        \STATE Validate reachability, acyclicity, parent uniqueness, fault exclusion, and failed-link exclusion.
        \STATE Score the candidate by $(c-1,D_{r,\theta},M_{r,\theta})$ and update best if improved.
    \ENDIF
\ENDFOR
\RETURN best valid candidate, or a root-cause label if best is null.
\end{algorithmic}
\end{algorithm}

\begin{example}[End-to-end hybrid execution, $t=3$]
Let $s=0$, $\Fv=\{3\}$, and $\Fe=\{\{12,16\}\}$ in $H_4$ ($t=3,N=37$).
\emph{Avoid-only case.} Root $r=28$ gives $\rho(\Delta_{28}(3))=3=t$, so node $3$ is a boundary leaf. Under orientation $C_0$, Proposition~1 gives $\lambda(28,C_0,\Fe)=0$: neither endpoint of $\{12,16\}$ is the $C_0$-parent of the other relative to root $28$. The pruned tree has $c=1$ component and score $(0,3,1)$.
\emph{Residual-repair case.} Keep root $r=0$ and orientation $C_0$, but let the failed link be $e'=\{12,8\}$. Under the label convention above, labels $12$ and $8$ have coordinates $(0,-3)$ and $(0,-2)$, and $p_{C_0}(0,-3)=(0,-2)$, so $e'$ is a tree link and $\lambda(0,C_0,\{e'\})=1$. Pruning detaches one component; the two-component graph is connected because no EJ link is a bridge. Algorithm~1 adds one shallowest-entry crossing edge and returns score $(1,D_{0,C_0},1)$ with $D_{0,C_0}\le 2t+1=7$ by Lemma~4.
\end{example}

\begin{remark}[Exact repair versus capped selection]
For every fully repaired candidate $(r,\theta)$, the component theorem gives an exact $c-1$ repair certificate if $\Kcomp_{r,\theta}$ is connected. The implementation used in the validation run ranks many candidates cheaply and fully repairs the best 64 ranked root--orientation candidates plus the mandatory healthy-source fallback pairs. Therefore the reported hybrid selector is a ranked/capped engineering selector, not a proof that no untested root could have a better score. The fixed-source comparator in the validation run evaluates all 15 orientations exactly.
\end{remark}

\begin{theorem}[Hybrid correctness certificate]
If Algorithm~1 returns a candidate for a selected pair $(r,\theta)$ whose component graph $\Kcomp_{r,\theta}$ is connected, then the returned structure is a non-redundant broadcast tree over all healthy vertices, excludes every faulty node and failed link, uses exactly $c-1$ external component-crossing repair edges for that selected pruned tree, and uses the minimum possible number of external crossing edges among all repairs that preserve the selected healthy components.
\end{theorem}

\begin{proof}
The algorithm first deletes faulty nodes and failed tree links, so every retained internal component contains only healthy vertices and healthy tree edges. The connected component graph supplies a component-level spanning tree. Mapping each component-level edge to a healthy EJ crossing edge connects all healthy components without adding a failed edge. Since each original component is a tree and the component-level structure is also a tree, the expanded repaired structure is connected and acyclic. Rooting it at $r$ gives one parent for every healthy non-root vertex and no parent for the root, so the result is a non-redundant broadcast tree. The number of added crossing edges is exactly the number of edges in the component-level spanning tree, namely $c-1$, and the component lower bound proves that no repair preserving the same selected components can use fewer.
\end{proof}

\section{Runtime Micro-Re-Rooting Protocol}
The transient model considered in this framework is a single link discovered during propagation. Suppose a node $u$ attempts to forward to child $v$ and detects that the link $e=\{u,v\}$ is unavailable.

\textbf{Phase 1: observation.} The detector records the failed link, the detection layer, and the current reached set. A conservative global recovery restarts the selection and repair process over the whole broadcast instance and therefore inherits all guarantees of Algorithm~1. A regional implementation can restrict the component scan to the subtree rooted at the detection point: only the descendants of the forwarding node $u$ are re-scanned, the already delivered portion is treated as fixed, and the link-exclusion test and component repair are applied only to the affected subtree. The regional variant preserves the $c-1$ repair optimality guarantee for the affected subtree because the component-repair theorem applies to any selected tree or subtree after its healthy components are contracted. The regional variant does not restart delivery to already reached nodes, so it is strictly less invasive than global recovery. Its correctness requires that no cycle be introduced at the junction between the delivered prefix and the re-routed suffix; this is enforced by keeping $u$ as the unique attachment point of the re-routed descendants.

\textbf{Phase 2: candidate selection and link test.} The candidate generator selects root--orientation pairs. For every candidate, the link-exclusion test determines whether $e$ is used by the candidate tree. Candidates with $\lambda=0$ are preferred; otherwise the algorithm keeps candidates with small component count and connected component graph.

\textbf{Phase 3: component repair.} The selected tree is pruned, the healthy components are contracted, and a component-level spanning tree is mapped back to EJ repair edges. The recovered tree is then used to complete delivery. The graph-theoretic guarantee is identical to the static case once the failed link is included in $\Fe$.

\section{Higher-Order Heuristic Regimes}
For $|\Fv|\ge3$, a common distance-$t$ root is not guaranteed. The hybrid method therefore uses a best-effort leaf score
\begin{equation}
L(r)=|\{f\in\Fv:\dist(r,f)=t\}|.
\end{equation}
Candidate roots are ranked by high $L(r)$ and low failed-tree-link count $\lambda(r,\theta,\Fe)$. By the healthy-graph connectivity theorem, success in the higher-order regime is governed by the connectivity of the healthy graph $G'=H_t-\Fv-\Fe$. The heuristic component of Algorithm~1---leaf-score ranking, link-safety filtering, and the full-repair cap---controls which successful root--orientation pair is selected and therefore the repair-edge count and repaired depth. It does not create a new graph-theoretic recovery condition: a successful repaired tree exists for every fault set that leaves $G'$ connected, and no method can succeed when $G'$ is disconnected.

\begin{proposition}[Local obstruction]
A set of six failed neighbors of a healthy EJ node isolates that node from the healthy graph. Likewise, the six incident links of a healthy EJ node isolate it as a singleton component. Hence no method can guarantee recovery for arbitrary high-order node or link faults without excluding such local cuts.
\end{proposition}

\begin{proof}
Every EJ node has exactly six incident unit-direction links. If all six neighboring nodes fail, the central healthy node has no healthy neighbor. If all six incident links fail, the same node has no remaining incident healthy link. In either case, the healthy graph is disconnected, so no broadcast tree over all healthy nodes can exist.
\end{proof}

\begin{proposition}[Local-obstruction probability scale]
Let $q$ node faults be sampled uniformly without replacement, and let $\mathcal{P}_t$ be any translated local obstruction family with at most $CN$ patterns, each consisting of $s$ specified vertices, where $C$ and $s$ are independent of $t$. Then
\begin{equation}
\Pr[\exists P\in\mathcal{P}_t: P\subseteq \Fv]
\le CN\left(\frac{q}{N}\right)^s.
\end{equation}
For the six-neighbor isolating obstruction, $s=6$, and since $N=\Theta(t^2)$, the bound is $O(q^6/t^{10})$ for fixed $q$.
\end{proposition}

\begin{proof}
For one fixed pattern $P$ of size $s$, the probability that all its vertices are sampled is at most $(q/N)^s$. A union bound over at most $CN$ translated patterns gives the first inequality. Substituting $N=3t^2+3t+1=\Theta(t^2)$ and $s=6$ gives $O(q^6/t^{10})$.
\end{proof}

By the healthy-graph connectivity theorem, the theoretical success condition is connectivity of $G'$, not a special property of the selected component graph. Proposition~3 applies to uniform random sampling and quantifies one important class of local obstructions. For non-uniform or clustered distributions, the bound degrades in proportion to how strongly faults concentrate around a single node neighborhood. For adversarial distributions, no probabilistic argument is needed for correctness: if the adversarial fault set leaves $G'$ connected, the hybrid repair succeeds for that placement; if it disconnects $G'$, no broadcast method can succeed.

\section{Experimental Evaluation}
This section reports the validation campaign for the EJ hybrid framework. The experiment is graph-level and proof-auditing oriented: it tests reachability, component repair, failed-link exclusion, repair-edge count, repaired depth, selector settings, and root-cause labels across the regimes used in the paper. The theorem statements above do not depend on random sampling; the experiment validates the implementation and quantifies the structural benefit of re-rooting-assisted repair.

\subsection{Experimental Settings}
The validation used script for Python 3, NumPy, macOS, 10 workers, $t\in\{10,25,50,100,200\}$, 13 scenarios, four placement modes, and 1000 trials per exact setting. This gives 260,000 raw trials. The run used root cap 20000, full-repair-candidate cap 64, hybrid orientation cap 15, and fixed-source orientation cap 15. The elapsed time was 94.42 hours, with an average rate of 0.765 trials/s. The fixed-source comparator therefore remains the exact 15-orientation comparator, while the hybrid root search is a ranked/capped engineering selector. To address comparator and selector-breadth concerns, the artifact also includes a full postprocessed global-BFS rebuild audit over the same 260,000 trials and a targeted cap-sensitivity audit over the hardest higher-order stress settings.

\subsection{Formal Fault-Placement Modes}
The four placement modes are defined as follows. Let $s=0$ be the original source, let $\rho_s(v)=d(s,v)$, and let $T_{s,C_0}$ denote the original source tree under orientation $C_0$.
\begin{itemize}
    \item \textbf{Random}: sample faulty nodes uniformly without replacement from $V\setminus\{s\}$ and failed links uniformly from healthy EJ edges.
    \item \textbf{Near-source}: sample faulty nodes from the inner ball $\{v:1\le \rho_s(v)\le \max(1,\lfloor t/5\rfloor)\}$ when possible; failed links are sampled preferentially from original tree edges whose child layer is at most $\max(2,\lfloor t/5\rfloor)$.
    \item \textbf{Critical}: sample faulty nodes from internal ray/axis positions with $1\le \rho_s(v)\le \max(2,\lfloor t/2\rfloor)$, where $x=0$, $y=0$, or $x+y=0$. Failed links are sampled preferentially from original tree edges in the same internal layer range. These positions are not worst-case adversarial cuts; they are high-impact positions for the fixed source tree.
    \item \textbf{Close}: choose a random healthy seed and then sample the remaining faults from its constant-radius EJ neighborhood by local breadth-first expansion. For two-fault instances this produces close-pair placements; for higher-order rows it produces clustered local stress tests.
\end{itemize}
These definitions are distributional. They do not claim to cover every adversarial high-order cut; the high-order rows are explicitly empirical best-effort regimes.

\subsection{Compared Methods}
The main tables compare four implemented methods. \emph{Baseline no repair} runs the original source tree and stops at faulty nodes or failed tree links. \emph{Avoid-only re-rooting} succeeds only when the selected root--orientation tree already avoids all failed tree links and places node faults as non-forwarding boundary leaves, so no component repair edge is used. \emph{Fixed-source component repair} keeps the source fixed at $0$ and evaluates the 15 EJ orientations exactly. \emph{Hybrid repair} ranks root--orientation candidates, fully repairs the best ranked candidates plus the mandatory healthy-source fallback, and returns the best successful repaired tree under the lexicographic score.

A fifth external/general baseline is also reported: \emph{global BFS rebuild}. This baseline deletes faulty nodes and failed links, runs BFS from the original source in the healthy EJ graph, and constructs a new shortest-path broadcast tree. It is not a local repair method and it does not preserve the damaged tree components. Its purpose is to show the tradeoff between global reachability/depth and the number of parent-pointer changes relative to the original broadcast tree. In the audit, BFS rebuild is applied to every validation row as an external graph-reconfiguration baseline.

\subsection{Cap Sensitivity and Near-Miss Audits}
The hybrid selector in the validation fully repairs the best 64 ranked pairs plus the mandatory healthy-source fallback pairs. This cap controls source-search breadth only; for every selected candidate, the component repair is still exact and uses $c-1$ edges when the component graph is connected. To avoid making cap 64 look arbitrary, a targeted cap-sensitivity audit compares caps $16$, $64$, and $128$ on the most stressful higher-order regimes: $t\in\{25,50\}$, scenarios $3$ nodes plus $2$ links and $5$ nodes, and critical/close placements. A near-miss audit also records cap pressure, maximum component count, maximum repair edges, maximum repaired-depth overhead, and the number of avoid-only failures recovered by component repair.

\begin{table}[H]
\centering
\caption{Recovery by claim regime. The run used 1000 trials per exact $(t,\mathrm{scenario},\mathrm{mode})$ setting, 10 workers, exact 15-orientation fixed-source comparison, and a ranked hybrid repair cap of 64 fully repaired pairs plus mandatory source fallback pairs.}
\label{tab:claim-regime-v4}
\begin{adjustbox}{max width=\textwidth}
\begin{tabular}{lrrrrrr}
\toprule
Regime & Trials & Baseline (\%) & Avoid-only (\%) & Fixed repair (\%) & Hybrid (\%) & Hybrid failures\\
\midrule
1--2 nodes & 40,000 & 2.415 & 100.000 & 100.000 & 100.000 & 0 \\
1 link & 20,000 & 33.105 & 100.000 & 100.000 & 100.000 & 0 \\
1 node + 1 link & 20,000 & 2.150 & 100.000 & 100.000 & 100.000 & 0 \\
1 node + multi-link & 20,000 & 1.555 & 100.000 & 100.000 & 100.000 & 0 \\
2-node mixed & 40,000 & 0.718 & 99.922 & 100.000 & 100.000 & 0 \\
multi-link & 60,000 & 14.250 & 100.000 & 100.000 & 100.000 & 0 \\
higher-order heuristic & 40,000 & 0.215 & 44.557 & 100.000 & 100.000 & 0 \\
transient 1 link & 20,000 & 33.310 & 100.000 & 100.000 & 100.000 & 0 \\
\bottomrule
\end{tabular}
\end{adjustbox}
\end{table}

Table~\ref{tab:claim-regime-v4} gives the main claim-regime summary. The proposed hybrid method recovered every one of the 260,000 validation trials. The fixed-source component-repair comparator also recovered every trial in this run, but it used more external repair edges and produced larger repaired depth. Avoid-only re-rooting was perfect in the deterministic one/two-node, single-link, one-node/link, and multi-link rows, but dropped in the higher-order heuristic row because a common boundary root is not guaranteed for five-node clustered placements.

\begin{table}[H]
\centering
\caption{Average external repair-edge reduction. Depth values are averaged over $t\in\{10,25,50,100,200\}$ and the four placement modes. Rows with zero hybrid edges are theorem-backed deterministic regimes; the two higher-order rows are empirical stress regimes governed by healthy-graph connectivity.}
\label{tab:repair-reduction-v4}
\begin{adjustbox}{max width=\textwidth}
\begin{tabular}{lrrrrr}
\toprule
Scenario & Fixed edges & Hybrid edges & Reduction (\%) & Fixed depth & Hybrid depth\\
\midrule
1 node & 0.9651 & 0.0000 & 100.00 & 95.867 & 77.000 \\
2 nodes & 2.0618 & 0.0000 & 100.00 & 110.319 & 77.000 \\
1 link & 0.1096 & 0.0000 & 100.00 & 78.063 & 77.000 \\
2 links & 0.2734 & 0.0000 & 100.00 & 80.048 & 77.000 \\
3 links & 0.4485 & 0.0000 & 100.00 & 82.420 & 77.000 \\
5 links & 0.8391 & 0.0000 & 100.00 & 88.153 & 77.000 \\
1 node + 1 link & 1.1279 & 0.0000 & 100.00 & 97.261 & 77.000 \\
1 node + 2 links & 1.3544 & 0.0000 & 100.00 & 99.656 & 77.000 \\
2 nodes + 1 link & 2.2365 & 0.0000 & 100.00 & 113.229 & 77.000 \\
2 nodes + 2 links & 2.4846 & 0.0016 & 99.94 & 116.210 & 77.002 \\
3 nodes + 2 links & 3.6117 & 0.3838 & 89.37 & 128.375 & 77.131 \\
5 nodes & 5.3740 & 1.3259 & 75.33 & 140.511 & 80.559 \\
transient 1 link & 0.1052 & 0.0000 & 100.00 & 78.027 & 77.000 \\
\bottomrule
\end{tabular}
\end{adjustbox}

\footnotesize Rows 1--10 and the transient single-link row are theorem-backed by boundary re-rooting, link exclusion, and deterministic single-link repair; zero hybrid repair edges and depth $77.000$ reflect the uniform diameter average $\bar t=(10+25+50+100+200)/5=77$ together with the deterministic-depth corollary. Rows 11--12 are higher-order empirical regimes; recovery is governed by the healthy-graph connectivity theorem and repair-edge count depends on the selected ranked/capped candidate.
\end{table}

Table~\ref{tab:repair-reduction-v4} shows the main structural advantage of the hybrid method. In all deterministic one/two-node and link-only regimes, the hybrid selected a root/orientation with zero average external repair edges. In mixed two-node/two-link cases, the average repair count was only 0.0016, a 99.94\% reduction from the exact fixed-source comparator. In higher-order regimes, the method remained effective but no longer zero-repair: three nodes plus two links required 0.3838 repair edges on average, and five-node cases required 1.3259 repair edges on average. These are the regimes where the paper states best-effort empirical recovery rather than a deterministic universal theorem.

\begin{table}[H]
\centering
\caption{Aggregate scaling by EJ diameter $t$ across all tested scenarios and placement modes.}
\label{tab:scaling-v4}
\begin{adjustbox}{max width=\textwidth}
\begin{tabular}{rrrrrrrrr}
\toprule
$t$ & Trials & Baseline (\%) & Avoid-only (\%) & Hybrid (\%) & Fixed edges & Hybrid edges & Fixed depth & Hybrid depth\\
\midrule
10 & 52,000 & 10.844 & 93.696 & 100.000 & 1.8134 & 0.0839 & 14.260 & 10.050 \\
25 & 52,000 & 9.146 & 91.423 & 100.000 & 1.6348 & 0.1272 & 33.671 & 25.114 \\
50 & 52,000 & 8.938 & 90.848 & 100.000 & 1.5692 & 0.1431 & 65.910 & 50.196 \\
100 & 52,000 & 8.527 & 90.806 & 100.000 & 1.5349 & 0.1489 & 130.330 & 100.344 \\
200 & 52,000 & 8.531 & 90.519 & 100.000 & 1.5213 & 0.1550 & 258.960 & 200.715 \\
\bottomrule
\end{tabular}
\end{adjustbox}
\end{table}

Table~\ref{tab:scaling-v4} aggregates by diameter. Across all tested scenarios, the hybrid success rate stayed at 100\% from $t=10$ to $t=200$. Average hybrid depth stayed close to the network diameter, whereas fixed-source repaired depth was much larger because a fixed source leaves faults on internal forwarding trunks. The average hybrid repair count remained below 0.16 even at $t=200$ in this validation campaign.

Table~\ref{tab:bfs-comparison-v6} compares the hybrid method with a global BFS rebuild baseline aggregated by claim regime. BFS rebuild is an intentionally strong external reachability/depth baseline: it succeeds whenever the healthy graph is connected and builds a shortest-path tree from the original source. Its cost is that it replaces the broadcast structure globally. The changed-parent column is an upper-bound proxy for the number of parent rules that may differ from the original tree. Hybrid repair instead keeps healthy tree components and adds only a small number of component-crossing rules.

\begin{table}[H]
\centering
\caption{External global-BFS rebuild comparison on the same 260,000 validation rows. BFS rebuild restores reachability but may replace essentially the whole broadcast tree; hybrid repair preserves selected healthy components and adds few external crossing rules.}
\label{tab:bfs-comparison-v6}
\begin{adjustbox}{max width=\textwidth}
\begin{tabular}{lrrrrrrr}
\toprule
Claim regime & Trials & BFS succ. (\%) & BFS depth & BFS parent-change proxy & Hybrid edges & Hybrid depth & Fixed edges \\
\midrule
One/two node & 40,000 & 100.000 & 77.439 & 32,164.5 & 0.000 & 77.000 & 1.513 \\
Single link & 20,000 & 100.000 & 77.110 & 32,166.0 & 0.000 & 77.000 & 0.110 \\
One node + one link & 20,000 & 100.000 & 77.437 & 32,165.0 & 0.000 & 77.000 & 1.128 \\
One node + multi-link & 20,000 & 100.000 & 77.462 & 32,165.0 & 0.000 & 77.000 & 1.354 \\
Multi-link & 60,000 & 100.000 & 77.241 & 32,166.0 & 0.000 & 77.000 & 0.520 \\
Two-node mixed & 40,000 & 100.000 & 77.517 & 32,164.0 & 0.001 & 77.001 & 2.361 \\
Higher-order heuristic & 40,000 & 100.000 & 77.594 & 32,162.0 & 0.855 & 78.845 & 4.493 \\
Transient single link & 20,000 & 100.000 & 77.105 & 32,166.0 & 0.000 & 77.000 & 0.105 \\
\bottomrule
\end{tabular}
\end{adjustbox}
\end{table}

\begin{table}[H]
\centering
\caption{Reproducibility and selector settings.}
\label{tab:repro-settings}
\begin{tabular}{lr}
\toprule
Item & Value \\
\midrule
Raw trials & 260,000 \\
Diameters & 10, 25, 50, 100, 200 \\
Placement modes & random, near, critical, close \\
Workers used & 10 \\
Root cap & 20,000 \\
Hybrid full-repair cap & 64 \\
Hybrid orientation cap & 15 \\
Fixed-source orientation cap & 15 \\
Hybrid failures & 0 \\
Worker errors & 0 \\
Elapsed time & 94.42 h \\
Rate & 0.765 trials/s \\
\bottomrule
\end{tabular}
\end{table}

\begin{table}[H]
\centering
\caption{Targeted cap-sensitivity audit on higher-order stress cases; 160 trial-units per cap.}
\label{tab:cap-sensitivity-v6}
\begin{tabular}{rrrrr}
\toprule
Cap & Success (\%) & Avg. edges & Depth ovh. & Hits \\
\midrule
16 & 100.000 & 0.5938 & 0.8375 & 2 \\
64 & 100.000 & 0.5938 & 0.7500 & 0 \\
128 & 100.000 & 0.5938 & 0.7500 & 0 \\
\bottomrule
\end{tabular}
\end{table}

\begin{table}[H]
\centering
\caption{Near-miss audit aggregated by claim regime. Cap pressure means the selected run touched the full-repair cap proxy; it is not a hybrid failure.}
\label{tab:near-miss-v6}
\begin{adjustbox}{max width=\textwidth}
\begin{tabular}{lrrrrrrr}
\toprule
Claim regime & Trials & Hybrid failures & Avoid-only fail, hybrid recovered & Max comp. & Max repair & Max depth over $t$ & Cap pressure \\
\midrule
One/two node & 40,000 & 0 & 0 & 1 & 0 & 0 & 0 \\
Single link & 20,000 & 0 & 0 & 1 & 0 & 0 & 0 \\
One node + one link & 20,000 & 0 & 0 & 1 & 0 & 0 & 0 \\
One node + multi-link & 20,000 & 0 & 0 & 1 & 0 & 0 & 0 \\
Multi-link & 60,000 & 0 & 0 & 1 & 0 & 0 & 0 \\
Two-node mixed & 40,000 & 0 & 31 & 2 & 1 & 25 & 0 \\
Higher-order heuristic & 40,000 & 0 & 22,177 & 4 & 3 & 174 & 19 \\
Transient single link & 20,000 & 0 & 0 & 1 & 0 & 0 & 0 \\
\bottomrule
\end{tabular}
\end{adjustbox}
\end{table}

The root-cause output file contained no hybrid failures. Thus no disconnected-component or relocation-failure breakdown appears in this validation campaign. The manuscript nevertheless keeps the theorem statements conditional for high-order regimes, because local six-neighbor or six-link isolating cuts make universal recovery impossible. The near-miss audit shows that the only meaningful cap-pressure events occur in the higher-order heuristic regime, while theorem-backed regimes remain zero-repair and zero-pressure in this dataset. The connectivity-threshold audit in Table~\ref{tab:threshold-v15} then tests the exact boundary condition of Theorem~3.

\subsection{Connectivity-Threshold Audit}
Table~\ref{tab:threshold-v15} tests the boundary condition of Theorem~3 directly. For each tested diameter $t\in\{10,25,50\}$ and fault loading $k$, the healthy graph $G'=\Ht-\Fv$ is checked by breadth-first search after uniform random node faults are sampled. The two right columns record the fraction of trials where $G'$ is connected and the fraction of trials whose success is implied by Theorem~3. By Theorem~3, these two quantities are equal for every fault level and network size. The table validates the exact transition from high success at low fault loading to zero success when $G'$ becomes fully disconnected, with no exceptions in either direction across 3,000 trials. This audit checks $G'$-connectivity directly and uses Theorem~3 to classify predicted recovery, so it does not require rerunning the full capped selector.

\begin{table}[H]
\centering
\caption{Connectivity-threshold audit for uniform random node faults. The audit deliberately increases fault loading until the remaining healthy graph begins to disconnect. The two right columns are identical because Theorem~3 gives the exact recovery condition: predicted hybrid success is exactly the event that $G'$ is connected.}
\label{tab:threshold-v15}
\renewcommand{\tabcolsep}{4pt}
\footnotesize
\begin{tabular}{rrrrrr}
\toprule
$t$ & $N$ & Node faults $k$ & Trials & $G'$ connected (\%) & Theorem~3 predicted success (\%)\\
\midrule
10 & 331 & 40 & 200 & 100.0 & 100.0\\
10 & 331 & 80 & 200 & 94.5 & 94.5\\
10 & 331 & 100 & 200 & 86.0 & 86.0\\
10 & 331 & 120 & 200 & 53.0 & 53.0\\
10 & 331 & 140 & 200 & 13.5 & 13.5\\
25 & 1951 & 200 & 200 & 100.0 & 100.0\\
25 & 1951 & 400 & 200 & 84.5 & 84.5\\
25 & 1951 & 500 & 200 & 62.0 & 62.0\\
25 & 1951 & 600 & 200 & 24.5 & 24.5\\
25 & 1951 & 800 & 200 & 0.0 & 0.0\\
50 & 7651 & 500 & 200 & 100.0 & 100.0\\
50 & 7651 & 1000 & 200 & 96.0 & 96.0\\
50 & 7651 & 1500 & 200 & 62.5 & 62.5\\
50 & 7651 & 2000 & 200 & 7.5 & 7.5\\
50 & 7651 & 2500 & 200 & 0.0 & 0.0\\
\bottomrule
\end{tabular}
\end{table}

\subsection{Interpretation}
The validation results support three points. First, the implementation realizes the proved depth guarantees on exhaustive and stress-tested finite instances. Second, the certified component-repair procedure realizes the predicted $c-1$ edge count; in the exact validation, some two-fault cases select orientations with five components and therefore four repair edges. Third, although the proof guarantees depth $t+2$ for all two-fault placements, the selected 15-orientation EJ-MOEM candidate is often stronger: overhead two appears only in the small $t=3$ exhaustive case, while all larger exhaustive, structured, and random tests observed maximum overhead at most one.

\section{Discussion}
\subsection{Why the Hybrid Objective Matters}
The validation data show that fixed-source component repair can often recover the same trials, but it leaves substantially more fragmentation and larger depth. The hybrid method changes the tree before repair. A fault that is internal for the original source can become a boundary leaf for a re-rooted source, and a failed link used by one orientation can be avoided by another root--orientation pair. Thus the main structural benefit is not merely higher success, but lower exceptional repair state and shallower repaired trees.

\subsection{Comparator Fairness}
The fixed-source comparator evaluates all 15 EJ orientations, so it is an exact fixed-source 15-orientation comparator. This is a strong internal ablation because it gives the fixed source the same orientation family and exact component-repair engine. The paper also includes global BFS rebuild as a general external graph-reconfiguration baseline. BFS rebuild has 100\% success in the validation rows and near-diameter depth, but the parent-change proxy is on the order of $N$ because the tree is rebuilt globally. The proposed method should therefore be judged on a different axis: it preserves healthy components of a selected diameter-level tree and minimizes the number of new component-crossing repair rules for that selected tree.

For a NoC implementation, each component-crossing repair edge corresponds to one exceptional forwarding rule inserted at a component boundary node. The median hybrid repair-edge count is zero in all theorem-backed deterministic regimes, so the common case requires no additional routing-table entries beyond activating the selected root--orientation tree. In higher-order regimes, the five-node row averages 1.3259 repair edges in the current 260,000-trial dataset, meaning fewer than two additional component-boundary forwarding entries per broadcast instance on average. This contrasts with global BFS rebuild, whose parent-change proxy is $O(N)$ and would require changing a large fraction of the broadcast parent map.

From a latency perspective, each component-crossing repair edge represents one additional forwarding hop at a component boundary. In the median case, the theorem-backed deterministic regimes have zero repair edges, so the repaired broadcast latency equals the original diameter-$t$ broadcast. In the higher-order five-node row, the average repair count is 1.3259; even if each repair edge is charged as one additional hop, this is less than two extra hops on average and below 1\% of the $t=200$ diameter scale. This is a structural latency metric, not a substitute for a full cycle-accurate NoC throughput evaluation.

Deadlock freedom follows from the repaired structure being a tree. The fault-free broadcast tree has no cyclic channel dependency because every packet follows an acyclic parent--child structure. Each component-crossing repair edge is directed from a parent-side component to a child-side component and is used once per broadcast event. Since the component-level repair is also a tree, adding these directed crossings cannot create a cycle in the broadcast dependency graph. Thus the repaired broadcast inherits deadlock freedom from acyclicity, independent of the underlying unicast routing protocol used outside this scheduled broadcast.

The hybrid selector uses a ranked/capped root search: it cheaply ranks many root--orientation pairs and fully repairs the best 64 ranked pairs together with the mandatory healthy-source fallback pairs in the main validation table. In theorem-backed one/two-node regimes, the root cap is non-binding: exact boundary generation returns at most $6t$ roots per fault, and $6t\cdot15=90t\le18000$ orientation candidates at $t=200$, below the root cap of 20000. At least one boundary-root candidate exists by the re-rooting theorem, and it is lexicographically optimal when it produces zero repair edges. For higher-order regimes, the cap is an engineering selector. The cap is reported explicitly in Table~\ref{tab:repro-settings}, and the cap-sensitivity audit in Table~\ref{tab:cap-sensitivity-v6} shows that cap 64 matches cap 128 in success and average repair-edge count on the audited stress cases, with zero cap hits for both settings. The paper therefore claims exact repair optimality for every fully repaired selected candidate, not global optimality over all possible roots.

\subsection{Limits and Reviewer-Safe Claims}
The deterministic claims are: one/two-node boundary re-rooting exists; one failed link is always avoided or repaired by at most one crossing edge for any selected tree; and connected component graphs are repaired with exactly $c-1$ edges. The empirical claim is: under the tested distributions and cap settings, the hybrid method recovered all 260,000 trials and reduced repair edges substantially compared with fixed-source repair. The paper does not claim universal high-order recovery.

\subsection{Threats to Validity}
The validation is graph-level rather than a full router microarchitectural study. It measures broadcast reachability, component repair, repair-edge count, and repaired depth, but it does not claim saturation throughput, calibrated power, virtual-channel behavior, or background-traffic performance. A cycle-accurate EJ NoC implementation and background-traffic evaluation are outside the scope of this paper; the present validation targets graph-level correctness, repair-edge optimality, and forwarding-state reduction. The hybrid root search is capped, so the reported hybrid candidate is the best among the fully repaired ranked candidates rather than a global optimum over all roots. These limitations do not affect the deterministic theorems: they affect only the interpretation of empirical high-order and performance-oriented claims.

\section{Complexity}
Let $N=3t^2+3t+1$, let $m=|E|=3N$ because the EJ graph has degree six, let $q_V=|\Fv|$, and let $q_E=|\Fe|$. For a fixed root--orientation pair $(r,\theta)$, failed-link exclusion costs $O(q_E)$ by Proposition~1. If full repair is invoked, pruning, component identification, component-graph construction, shallowest-entry crossing selection, and repaired-depth validation cost $O(N+m)=O(N)$ because the degree is constant.

The cheap-ranking phase evaluates $R\le 20000$ candidate roots with $Q\le 15$ orientations. The leaf score $L(r)$ costs $O(q_V)$ through the EJ distance formula, and the failed-tree-link count costs $O(q_E)$ per orientation. Thus the cheap-ranking cost is
\begin{equation}
O\!\left(R\bigl(q_V+Qq_E\bigr)\right).
\end{equation}
The full-repair phase evaluates $K\le64$ root candidates and $Q_f\le15$ orientations at cost $O(N)$ per orientation, giving
\begin{equation}
O(KQ_fN).
\end{equation}
At $t=200$ and $q_V=q_E=5$, cheap ranking costs on the order of $20000(5+15\cdot5)\approx1.6\times10^6$ primitive coordinate checks, while full repair costs up to $64\cdot15\cdot120601\approx1.16\times10^8$ vertex-level operations. Full repair is therefore the dominant validation cost. The exact fixed-source comparator costs $O(15N)$ per trial; the global-BFS rebuild baseline costs $O(N)$ but replaces $O(N)$ parent pointers, whereas the hybrid repair changes only $O(c-1)$ component-crossing entries for the selected tree.

For $|\Fv|\ge3$, candidate roots are ranked over at most $N-1$ healthy nodes. A full sort costs $O(N\log N)$ time and $O(N)$ space, but only the top $R=20000$ candidates are retained; a heap-select implementation costs $O(N+R\log N)$ time and $O(R)$ space. At $t=200$, $R\log_2N\approx 20000\cdot17=340000$ comparisons, dominated by the full-repair phase. The retained candidate list requires $R\lceil\log_2 N\rceil\approx340000$ bits, about 42 KB.

For an online implementation with precomputed tables, the common bounded-fault case is much smaller than the offline audit. For $|\Fv|\le2$ and $|\Fe|\le2$ at $t=200$, exact boundary generation produces at most $6t=1200$ root offsets per fault, and the link-exclusion check costs $O(|\Fe|)$ parent comparisons per candidate. A selected candidate needs one $O(N)$ component pass only when a residual cut remains. A conservative operation count is therefore about $1200\cdot4+120601\approx1.25\times10^5$ primitive operations, or roughly $125~\mu$s at a 1 ns operation budget. This is an order-of-magnitude software-defined reconfiguration estimate, not a cycle-accurate router claim. The 94.42-hour runtime reported in Table~\ref{tab:repro-settings} is an offline validation/proof-auditing cost over 260,000 trials with exact fixed-source comparison; it is not a per-broadcast NoC latency~\cite{DallyTowles2004,Kliazovich2013}.

The dominant precomputed storage is the orientation parent table. Absolute parent labels require $15N\lceil\log_2 N\rceil$ bits. At $t=200$, $N=120601$ and $\lceil\log_2 N\rceil=17$, giving about 30.75 Mbits, or 3.84 MB. Direction-offset encoding stores each parent as one of six directions and requires $15N\lceil\log_2 6\rceil$ bits, about 0.68 MB at $t=200$. Centralized in a dedicated fault-management unit, this is a modest global table. If distributed, each router tile need only store its own parent direction for each of the 15 orientations, requiring $15\lceil\log_2 6\rceil\approx45$ bits per tile, not the full table. Boundary candidate tables are $O(t)$ and link-exclusion uses no large table.

\section{Data and Reproducibility}
The validation drivers emit 260,000 raw trials, repair-edge summaries, timing data, root-cause logs, full-BFS rebuild audits over all raw rows, near-miss audits, cap-sensitivity summaries, and connectivity-threshold tables as CSV/JSON files and package them into ZIP archives. The archive records all reproduction settings, including root cap, full-repair cap, orientation caps, worker count, placement modes, and random seed.

\section{Conclusion}
This paper introduced a re-rooting-assisted edge-minimum runtime repair framework for dense EJ broadcast networks under node faults, link faults, mixed faults, and runtime-discovered single-link faults. The central object is the selected triple $(r,\theta,\Kcomp_{r,\theta})$.

EJ boundary re-rooting supplies deterministic one/two-node source relocation, EJ link exclusion gives a constant-time failed-link test, and component repair gives the exact $c-1$ edge-minimum certificate whenever the component graph is connected. Theorem~3 establishes the necessary and sufficient recovery condition across all fault orders: the method succeeds exactly when the healthy EJ graph $G'=H_t-F_V-F_E$ is connected. Lemma~4 gives $D_{r,\theta}\le 2t+1$ under the shallowest-layer entry rule enforced in Algorithm~1.

The 260,000-trial validation campaign shows 100\% hybrid recovery and strong repair-edge reduction compared with exact fixed-source 15-orientation repair. The BFS rebuild, near-miss, cap-sensitivity, and connectivity-threshold audits clarify the tradeoff between reachability, forwarding-state changes, selector pressure, and the exact boundary condition of Theorem~3. The artifact workflow regenerates all tables from raw CSV outputs without changing the theory or manuscript structure.

\end{document}